\documentclass{article}

\usepackage[english]{babel}
\usepackage{latexsym}
\usepackage{amssymb,amsmath,amsthm}
\usepackage{cite}    
\usepackage{ifthen}
\usepackage{verbatim}  
\usepackage{xspace}
\usepackage{paralist}
\usepackage{tikz}
\usepackage[colorlinks,final,citecolor=blue]{hyperref}

\usetikzlibrary{%
  arrows,%
  positioning,%
  decorations.pathmorphing,%
  decorations.pathreplacing,%
  automata,%
}
\newenvironment{transducer}[1][]
	{\begin{center}
	\begin{tikzpicture}[font=\footnotesize,shorten >=1pt,node distance=3.5cm,
	on grid,>=stealth',initial text=,every node/.style={align=center},
	every state/.style={inner sep=1pt},
	every path/.style={->,bend angle=20},#1]}
	{\end{tikzpicture}\end{center}}

\newcommand\pssi{\par\smallskip\indent}
\newcommand\pmsi{\par\medskip\indent}
\newcommand\pbsi{\par\bigskip\indent}
\newcommand\pssn{\par\smallskip\noindent}
\newcommand\pmsn{\par\medskip\noindent}
\newcommand\pbsn{\par\bigskip\noindent}
\newcommand\pnsi{\par\indent}
\newcommand\pnsn{\par\noindent}

\newcommand\emdef[1]{\texttt{#1}}  
\newcommand\emshort[1]{\emph{\textbf{#1}}}

\newcommand{\sk}[1]{\pmsi{\color{blue}SK $-->$ #1 !!!}\pmsn}

\theoremstyle{plain}
\newtheorem{theorem}{Theorem}[]

\newtheorem{lemma}[theorem]{Lemma}

\theoremstyle{definition}
\newtheorem{definition}[theorem]{Definition}
\newtheorem{example}[theorem]{Example}
\theoremstyle{remark}
\newtheorem{remark}[theorem]{Remark}

\newcommand{\cP}{\ensuremath{\mathcal{P}}\xspace}

\renewcommand{\phi}{\varphi}
\renewcommand{\epsilon}{\varepsilon}

\newcommand{\sse}{\subseteq}

\newcommand{\es}{\emptyset}
\newcommand{\sm}{-}  
\newcommand{\cupdot}{\mathbin{\mathaccent\cdot\cup}}
\newcommand\card[1]{|#1|}
\newcommand\pset[1]{2^{#1}}  

\newcommand{\N}{\mathbb{N}}

\newcommand{\e}{\lambda}
\newcommand{\ew}{\e}      
\newcommand\al{\Sigma}        
\newcommand\mxl{\ensuremath{\mathbf{M}}}
\newcommand{\ord}{\prec}
\newcommand{\pos}{\mathrm{pos}}

\newcommand{\mr}[1]{#1^\sim}

\newcommand\pty{\cP}  
\newcommand\ptyt{\pty_{\trt}}   

\newcommand\tr{\textbf{t}}       
\newcommand{\trt}{\textbf{t}}    
\newcommand\trx{\mathbf{x}}   
\newcommand{\trti}{\textbf{t}^{-1}}    
\newcommand\tbx{\mathbf{b}}   
\newcommand\tpx{\mathbf{p}}   
\newcommand\tsx{\mathbf{s}}   
\newcommand\tsub[1]{\mathbf{sub}_{#1}}   
\newcommand\tdsub[1]{\mathbf{sub}^{\ord}_{#1}}   
\newcommand\tdsubi[1]{(\mathbf{sub}^{\ord}_{#1})^{-1}}   
\newcommand\tdid[1]{\mathbf{id}^{\ord}_{#1}}   
\newcommand\trs{\textbf{s}}      
\newcommand\rel{\ensuremath\mathrm{R}\xspace} 
\newcommand\chann{\Gamma}        
\mathchardef\mhyphen="2D
\newcommand\chid[1] {\chann_{#1\mhyphen\mathrm{id}}} 
\newcommand\chsid[1]{\chann_{#1\mhyphen\mathrm{sid}}} 
\newcommand\chsub[1]{\chann_{#1\mhyphen\mathrm{s}}} 
\newcommand\chins[1]{\chann_{#1\mhyphen\mathrm{i}}} 
\newcommand\chdel[1]{\chann_{#1\mhyphen\mathrm{d}}} 
\newcommand\inin{\downarrow}  
\newcommand\outin{\uparrow}   
\newcommand\mop{\mu}  
\newcommand\diff[2]{\Delta^{#1}_{#2}}
\newcommand\ind{\mathrm{I}}  
\newcommand\lop{\mathrm{Op}}  
\newcommand\invt{\sigma_X}   
\newcommand\invtY{\sigma_Y}   
\newcommand\invtK{\sigma_K}   
\newcommand\invtL{\sigma_L}   
\newcommand\invtx[1]{\sigma_{#1}}   

\begin{document}

\begin{center}
\textbf{\Large Embedding rationally independent languages into maximal ones 
}
%
\pbsn
{\large Stavros Konstantinidis and Mitja Mastnak}
\pmsn
Department of Mathematics and Computing Science\\
Saint Mary's University, Halifax, Nova Scotia, B3H 3C3, Canada\\
\texttt{s.konstantinidis@smu.ca, mmastnak@cs.smu.ca
}
\end{center}

\pbsn
\textbf{Abstract.}
We consider the embedding problem in coding theory:
given an independence (a code-related property) 
and an independent language $L$,
find a maximal independent language containing $L$.
We consider the case where the code-related 
property is defined via a rational binary relation
that is decreasing with respect to any fixed
total order on the set of words. Our method 
works by iterating a max-min operator that has
been used before for the embedding problem for
properties defined by length-increasing-and-transitive 
binary relations. By going to order-decreasing 
rational relations, represented by input-decreasing transducers, we are able to include many
known properties from both the noiseless and noisy domains of coding theory, as well as  any 
combination of such properties. Moreover, 
in many cases the desired maximal embedding is 
effectively computable.

\pbsn
\textbf{Keywords}.
codes, embedding, error control codes, independence, languages, maximal, transducers, variable-length codes

\section{Introduction}\label{sec:intro}
The embedding problem for a language $L$ satisfying a property  
$\pty$ is to find a language $L'$ that contains $L$ and is maximal satisfying $\pty$. This problem is meaningful when the property $\pty$ is an independence. In particular, many natural
code-related properties are independences with respect to  binary  relations on words. In this setting, a binary relation $\rho$ defines the property that consists of all languages
in which no two different words are related via $\rho$. 
Such languages are called $\rho$-independent.
The embedding problem has been addressed well for properties defined by length-increasing-and-transitive relations \cite{VHH:2005}, as well as for several fixed properties like
the bifix code property \cite{ZhSh:1995}, the solid code 
property \cite{Lam:2001}, and the bounded deciphering delay
property
\cite{Bruyere:91}. In \cite{DudKon:2012}, the authors consider 
properties where the relation $\rho$ is rational and, therefore, described by a finite transducer $\trt$. In this 
setting, assuming the given language $L$ is regular, one can
decide whether $L$ is a maximal $\trt$-independent language.
The contributions of the present paper are as follows.
\begin{itemize}
\item 
We introduce the concept of input-decreasing transducer,
which realizes order-decreasing relations, as a tool for defining many natural code-related properties, including variable-length code properties and error-detection properties, as well as any combinations of those. Assuming a fixed, but 
arbitrary total order 
on words, an input-decreasing transducer $\trt$ is such that, for any input word $w$, all output words of $\trt$ have a (strictly) smaller order than $w$. 
\item
We show that starting with any $\trt$-independent language $L$,
we can embed $L$ into  a maximal $\trt$-independent language $\mop_\trt^*L$, by iterating the max-min operator $\mop_\trt$
on $L$. The non-iterated operator $\mop_{\trt}$ is 
considered in \cite{VHH:2005} where it is shown that if $\trt$ is length-decreasing
and transitive, then any $\trt$-independent language $L$ 
is embedded into $\mop_\trt L$ which is maximal $\trt$-independent. In many cases, $\mop^*$ converges after finitely many steps. We also show a natural example of a $\trt$ where
$\mop^*_{\trt}$ does \emshort{not} converge after finitely many steps.
\item
Our embedding results hold for any fixed, but arbitrary, language $\mxl$ 
relative to which maximality is considered, that is, we embed any $L\sse\mxl$ into a maximal $\trt$-independent subset of $\mxl$---this idea of relative maximality has been considered before, e.g., in \cite{Restivo1990,DJKM:2012,DudKon:2012}. 
When $\mxl$ is finite, $\mop_\trt$ always converges after 
$i$ iterations, for some $i$, to $\mop_\trt^i L$. When both
$\mxl$ and $L$ are regular and $\mop_\trt$ converges after
finitely many operations, then $\mop_\trt^* L$ is computable.
With our approach we provide a solution to the embedding problem for 
many classical cases of both variable-length codes
(for $\mxl$ = all possible words) and 
error-detecting codes for 
substitution and for synchronization types of errors
(for $\mxl$ = all words of a certain length).
\end{itemize}
The paper is organized as follows. The next section contains information about the basic notation and terminology
used in the paper, and Section~\ref{sec:codes} 
provides some background information on independent languages and maximal embeddings, and introduces the iterated max-min operator.
Section~\ref{sec:smooth} contains a few technical results
and the weak condition  of a  transducer being smooth,
which guarantees that when the max-min operator converges 
in finitely many iterations, then it produces a maximal embedding. Section~\ref{sec:decr} focuses on input-decreasing transducers, which are always smooth and guarantee
that the iterated max-min operator  produces a
maximal embedding. Section~\ref{sec:ex}
demonstrates with several examples that the
concept of input-decreasing transducer can be
used to define many known properties from both the 
noiseless and noisy domains of coding theory.
In that section we also show an example of an input-decreasing transducer for which the iterated operator does not converge finitely.
Finally, the last section contains a few concluding remarks and directions for future research.

\section{Basic notions and notation}\label{sec:two}

In this section we present our notation
and terminology about words, languages, transducers
and word operators. 
\pssn

We write $\N,\N_0$ for the sets of natural numbers (not including 0) and non-negative
integers, 
respectively.
If $S$ is a set, then $\card S$ denotes the cardinality of $S$,
and $\pset S$ denotes the set of all subsets of $S$.
An \emdef{alphabet} is a finite nonempty set of symbols. In this paper,
we write $\al$ for any arbitrary alphabet.
The set of all words, or strings, over $\al$ is written as
$\al^*$ and includes the \emdef{empty word} $\ew$.
A \emdef{language} (over $\al$) is any set of words.
In the rest of this paragraph, we use the following arbitrary object names: $i,j$ for nonnegative integers, $K,L$ for languages and $u,v,w,x,y$ for words.
If $w\in L$ then we say that $w$ is an \emdef{$L$-word}.
When there is no risk of confusion, we write a singleton language $\{w\}$
simply as $w$. For example,
$L\cup w$ and $v\cup w$  mean $L\cup\{w\}$ and $\{v\}\cup\{w\}$, respectively.
We use standard operations and notation on words and languages
\cite{Wood:theory:of:comput,MaSa:handbook,FLHandbookI}.
For example, $|w|$, $uv$, $w^i$, 
$KL$, $L^i$, $L^*$, $L^+$ denote respectively, the length of $w$, the concatenation of $u$ and $v$, the word consisting of $i$ copies of $w$, 
the concatenation of $K$ and $L$, the language consisting of all words obtained by
concatenating any $i$ $L$-words, the Kleene star of $L$, and $L^+=L^*\setminus\ew$. If $w$ is of the form $uv$ then $u$ is a
\emdef{prefix} and
$v$ is a \emdef{suffix} of $w$. If $w$ is of the form $uxv$ then $x$ is an
\emdef{infix} of $w$. If $u\not=w$ then $u$ is called a \emdef{proper prefix} of $w$---the definitions of proper suffix and proper infix are similar.
\pmsn
\emshort{Transducers and (word) relations \cite{Be:1979,Yu:handbook,Sak:2009}.}
A (word) \emdef{relation} over $\al$ 
is a subset
of $\al^*\times\al^*$, that is, a set of pairs $(x,y)$ of words
over the alphabet. The \emdef{inverse} of a relation $\rho$,
denoted by $\rho^{-1}$, is the relation 
$\{(y,x)\mid (x,y)\in \rho\}$.
The relation is \emdef{transitive} if $(x,y),(y,z)\in\rho$ implies $(x,z)\in\rho$, for all words $x,y,z$; that is, 
$\rho\circ\rho\sse\rho$, where `$\circ$' denotes composition.
Following \cite{VHH:2005}, the relation is called \emdef{length-increasing} (resp. length-decreasing) if
$(x,y)\in\rho$ implies $|x|<|y|$ (resp. $|x|>|y|$).
\pssn
A (finite) \emdef{transducer} is a quintuple 
$\tr=(Q, \al, T, I, F)$
such that $Q$ is the set of states, $I,F\sse Q$ are the sets of initial and final states, respectively, 
$\al$ is  the alphabet
and $T\subseteq Q\times\al^*\times\al^*\times Q$ is the finite set of transitions. Note that, in general transducers, one considers an input and an output alphabet, but in this paper the input and output alphabets are the same.
The \emdef{relation realized}
by the transducer $\tr$, denoted by $\rel(\tr)$, is the set of labels in all the accepting paths of $\tr$.
We write $\tr(x)$ for the set of \emdef{possible outputs of} $\tr$ on input $x$, that is,  $y\in\tr(x)$ iff $(x,y)\in \rel(\tr)$.
This notation is extended naturally to any language $X$:
\[
\trt(X) = \bigcup_{x\in X}\trt(x).
\]
The \emdef{inverse} of a transducer $\tr$, denoted by $\tr^{-1}$, is the transducer that results from $\tr$ by simply switching  
the input and output parts of the
labels in the transitions of $\tr$. It follows that $\tr^{-1}$ realizes the inverse of
the relation realized by $\tr$. 
If $\trt$ and $\trs$ are
transducers, then there are (effectively) a transducer 
$(\trt\lor\trs)$ 
realizing $\rel(\trt)\cup\rel(\trs)$ and a transducer 
$(\trt\circ\trs)$ realizing $\rel(\trt)\circ\rel(\trs)$.
By composing a transducer with itself $i$ times, for $i\in\N$,
we obtain a transducer which we denote by $\trt^i$. We define
\[
\trt^{-i}  \triangleq (\trt^{-1})^i.
\]
\begin{remark}\label{rem1}
For all $i\in\N$, we have that
\[
\rel(\trt^{-i}) = (\rel(\trt^i))^{-1}.
\]
Indeed, note that $(y,x)\in\rel(\trt^{-i})$
$\Leftrightarrow$ ``there are words $x_1,\ldots,x_{i-1}$ such that $(y,x_1)\in\rel(\trt^{-1})$, 
$(x_1,x_2)\in\rel(\trt^{-1}),\ldots,$ $(x_{i-1},x)\in\rel(\trt^{-1})$'' $\Leftrightarrow$ ``there are words $x_1,\ldots,x_{i-1}$ such that $(x,x_{i-1})\in\rel(\trt),\ldots,$ 
$(x_2,x_1)\in\rel(\trt),$ $(x_1,y)\in\rel(\trt)$'' $\Leftrightarrow$ $(x,y)\in\rel(\trt^i)$ $\Leftrightarrow$ 
$(y,x)\in(\rel(\trt^i))^{-1}$.
\end{remark}
A transducer $\trt$ is \emdef{transitive} if $\rel(\trt)$ is
transitive, that is, $\rel(\trt^2)\sse\rel(\trt)$.
For any regular language $L$, the relations $\rel(\tr)\cap (\al^*\times L)$ 
and $\rel(\tr)\cap (L\times\al^*)$ are
regular. The details of a transducer realizing $\rel(\tr)\cap (\al^*\times L)$, denoted
by  $\tr\outin L$, and of a transducer realizing $\rel(\tr)\cap (L\times\al^*)$, denoted
by $\tr\inin L$, are shown in \cite{Kon:2002}; thus,
\[
u\in(\tr\outin L)(w) \quad\hbox{ if and only if } \quad u\in\tr(w) \hbox{ and } u\in L.
\]
%
%
\pmsn\emshort{Language operators.}
A language operator is a function  
$\lop:2^{\al^*}\to2^{\al^*}$.
If $\lop$ is any language operator, $X$ is any language 
and $i$ is any nonnegative integer,
then we can define the following language operators.
\pbsi
$\lop^0(X)=X$ and $\lop^{i+1}(X)=\lop(\lop^{i}(X))$
\pbsi
$\lop^{\le i}(X) = X\cup\lop(X)\cup\cdots\cup\lop^i(X)$
\pbsi
$\lop^{\ge i}(X) = \lop^i(X)\cup\lop^{i+1}(X)\cup\cdots$
\pbsi
$\lop^{*}(X)=\cup_{i=0}^{\infty}\lop^i(X),\>\>\>$
$\lop^{+}(X)=\cup_{i=1}^{\infty}\lop^i(X)$
\pbsi
$\lop^{\cap}(X)=\cap_{i=1}^{\infty}\lop^i(X)$
\pmsn
If $\lop_1$ is also a language operator then we write
\[
\lop\sse\lop_1
\]
to indicate that $\lop(X)\sse\lop_1(X)$ for all languages $X$.
\pssn
We view a transducer $\trt$ as a language operator, so
the expressions $\trt^*$ and $\trt^\cap$, for instance, are legitimate in this paper. With this convention we can say that
a transducer is transitive if and only if
\[
\trt^2\sse\trt.
\]
For transducer operators  we also have that 
$\trt(\cup_i X_i)=\cup_i\trt(X_i)$, for all language families $(X_i)_i$. Using the above notation for language operators and Remark~\ref{rem1} we have the following.
\begin{remark}\label{rem2}
If the transducer $\trt$ is transitive then $\trt^+=\trt$
and $\trt^{-1}$ is also transitive.
\end{remark}
\section{Codes and the max-min operator}\label{sec:codes}
Here we provide background information on
code-related properties (independence properties)
and introduce the iterated max-min operator that is used to
embed a given independent language to a maximal one.
A \emdef{property} (over $\al$) is any set $\pty$ of languages.
If $L$ is in $\pty$ then we say that $L$ \emdef{satisfies} $\pty$.
A \emdef{code property}, or \emdef{independence}, \cite{JuKo:handbook},
is a property $\pty$ for which there is $n\in\N\cup\{\aleph_0\}$
such that
\[
L\in\pty, \quad\hbox{if and only if} \quad
L'\in\pty,\hbox{ for all $L'\sse L$ with $0<\card{L'}<n$,}
\]
that is, $L$ satisfies the property exactly when all nonempty subsets of $L$ with
less than $n$ elements satisfy the property. In the rest of the paper
we only consider properties $\pty$  that are independences.
A language $L\in\pty$  is called \emdef{$\pty$-maximal}, or a
maximal $\pty$ code, if $L\cup w\notin\pty$ for any word $w\notin L$.
From \cite{JuKo:handbook} we have that every $L$ satisfying $\pty$ is
included in a maximal $\pty$ code. To our knowledge, with possibly very few exceptions, all known code related 
properties in the
literature \cite{Shyr:book,JuKo:handbook,SSYu:book,LinCo:2004,Dom:2004,BePeRe:2009,MBT:2010,DudKon:2012}
are code properties as defined above. In this work we focus on input-altering transducer
properties. A transducer $\trt$ is called \emdef{input-altering} if
$
w\notin\trt(w),\>\>\hbox{ for all words $w$}.
$
A language $L$ is called $\trt$-\emdef{independent}
if 
\begin{equation}\label{eqPt}
\trt(L)\cap L=\es.
\end{equation}
The independence $\ptyt$ \emdef{described by} $\trt$ is the set of all $\trt$-independent languages.
It is easy to verify that the above equation is equivalent to
\begin{equation}\label{eqPtinv}
\trt^{-1}(L)\cap L=\es
\end{equation}
and also to
\begin{equation}\label{eqPtboth}
(\trt^{-1}\lor\trt)(L)\cap L=\es
\end{equation}
Thus, any of $\trt$, $\trt^{-1}$, $\trt\lor\trt^{-1}$ can be used to describe the same code property.
\begin{remark}
Let $\trt$ be an input-altering transducer.
Every singleton language $\{w\}$ is $\trt$-independent.
\end{remark}
\begin{remark}
The approach of input-altering transducers   
constitutes a realization in algorithmic
terms of independences defined via binary relations and
includes many known properties such as prefix codes,
bifix codes, outfix codes, and many error-detecting languages, as well as all the intersections of any two such properties.
In particular, for any binary relation $\rho$, 
a language $L$ is \emdef{$\rho$-independent} if 
\begin{equation}\label{eqPrelation}
u,v\in L\>\hbox{ and }\> (u,v)\in\rho
\> \hbox{ implies } \> u=v.
\end{equation}
The above statement implies that $\rho$-independence
is the same as $\rho^{-1}$ independence.
Let $\rho_{\not=}=\{(x,y)\in\rho\mid x\not= y\}$.
If $\rho_{\not=}$ is rational then there is an
input-altering transducer $\trt$ realizing it,
and condition~(\ref{eqPrelation}) is equivalent to any of
(\ref{eqPt})---(\ref{eqPtboth}) above.
The representation of code properties by transducers
(or other formal objects such as trajectories \cite{Dom:2004}) has lead to the implementation of a package for
manipulating objects representing code properties
\cite{Fado}, as well as to an online tool for  answering questions about code properties \cite{Laser}. 
\end{remark}

In the rest of the paper we consider a fixed, but arbitrary, input-altering transducer \trt, and a fixed, but arbitrary, language \mxl. 
Let $X$ be any language.  We define the following language operators.
\pbsi
$\ind_\trt(X) = \mxl\sm (\trt(X)\cup \trt^{-1}(X))$
\quad and \quad
$\mop_\trt(X) = \ind_\trt(X)\sm\trt^{-1}(\ind_\trt(X))$
\pmsn
When the transducer $\trt$ is understood, we omit above the subscript $\trt$. Also, as the operator $\mop$ is used heavily, 
we usually omit parentheses when applying $\mop$ on a language
$X$.
So the two operators are also written, respectively, as 
\pbsi
$\ind(X) = \mxl\sm (\trt(X)\cup \trt^{-1}(X))$
\quad and \quad
$\mop X = \mop(X)=\ind(X)\sm\trt^{-1}(\ind(X))$
\pmsn
The above operators are essentially translated to our transducer notation from the corresponding ones in
\cite{VHH:2005}. 
The operator $\ind(\cdot)$ is the set of all possible
words that are either in $X$ or $\trt$-independent
from $X$, so in some sense it is the \emdef{max}imum set
in which $X$ can be embedded. However, two words
in $\ind(X)\sm X$ might be $\trt$-dependent.
The operator mapping any $Y$ to $Y\sm\trti(Y)$ is
the `$\trt$-\emdef{min}imize' operator which returns all
$Y$-elements that cannot produce another $Y$-element
via $\trt$. The term `minimize' makes sense in our context
of input-decreasing transducers further below.
\begin{definition}
The operator $\mop_\trt$, or simply $\mop$ when $\trt$  is understood,  shown above is called
the \emdef{max-min operator}. The operator $\mop^*$ is 
called the \emdef{iterated} max-min operator. We say that it \emdef{converges finitely on a language} $L$, if there is $i\in\N_0$ such that
$\mop^*L=\mop^iL$.
\end{definition}
In the case of codes defined by length-increasing-and-transitive relations (equivalently, length-decreasing-and-transitive relations),
already the language $\mop L$ is maximal and constitutes a solution to the embedding problem, where
$L$ is the given language satisfying the code property. As stated in \cite{VHH:2005}, however, this
does not work for other codes like bifix codes, and also for error-detecting codes. A main observation in
this paper is that for any
$\trt$-independent language $L$, the language
$\mop^*L$ is an embedding of $L$, provided that 
$\trt$ satisfies a reasonable condition---see 
Section~\ref{sec:decr}.
%
%
%
\section{Smooth Transducer Operators}\label{sec:smooth}
In this section we
obtain several technical results about the max-min operator $\mop$ and we demonstrate Theorem~\ref{thGeneral}, which
states that when $\trt$ is smooth and $\mop^*$ converges  finitely  on some initial $\trt$-independent language $L$, then the resulting language is a $\trt$-independent maximal embedding of $L$. The concept of a 
smooth transducer is rather technical and is intended to
keep the results general. 
All input-decreasing transducers  of the next section are smooth.
\pssn
The second statement of the next lemma is the analogue of a statement in \cite{Dom:2004} 
concerning codes defined via trajectories.
\begin{lemma}\label{lemXYL}
Let $X,Y$ be any languages and let $L$ be a language satisfying the property $\ptyt$. The following statements hold true.
\begin{enumerate}
\item
If $X\sse Y$ then $\ind(Y)\sse\ind(X)$.
\item
$X\sm\trt^{-1}(X)$ satisfies $\ptyt$.
\item
$\mop^i L$ satisfies $\ptyt$ and $L\sse\mop^i L\sse\mop^{i+1}L\sse\ind(L)$, for all $i\in\N$.
\end{enumerate}
\end{lemma}
\begin{proof}
The first statement follows from the definition of $\ind$.
For the \underline{second} statement, we need to show that  Eq.~(\ref{eqPt}) holds for 
$L=X\sm\trt^{-1}(X)$. 
For the sake of contradiction assume that there is $w\in X\sm\trt^{-1}(X)$ and
$w\in\trt(X\sm\trt^{-1}(X))$.  Then $w\in\trt(u)$ for some $u\in X\sm\trt^{-1}(X)$, which
implies $u\notin\trt^{-1}(w)$ and, then $w\notin\trt(u)$, which is impossible.
\pnsi
For the \underline{third} statement, we first show that for any language $K$ satisfying $\ptyt$, we have
\begin{equation}\label{eqK}
K\sse\mop K\sse\ind(K)\>\hbox{ and }\>\mop K \hbox{ satisfies $\ptyt$}.
\end{equation}
The previous statement of the lemma implies that indeed $\mop K$ satisfies $\ptyt$.
The definition of $\mop$ implies that $\mop K\sse\ind(K)$.
Now, as  $K$ satisfies both Eq~(\ref{eqPt}) and (\ref{eqPtinv}),
we have that $K\cap(\trt(K)\cup\trt^{-1}(K))=\es$ and, therefore, $K\sse\ind(K)$.
If it were the case that $K\cap \trt^{-1}(\ind(K))\not=\es$, then also
$\trt(K)\cap\ind(K)\not=\es$, which is impossible. Hence, $K\sse\mop K$.
Now the statement follows if we use $L$ or $\mop^i(L)$ in place of $K$ in~(\ref{eqK}),
taking also into account the first statement of the lemma.
\end{proof}
In going from the length-increasing-and-transitive binary relations of \cite{VHH:2005} to the
input-altering ones of \cite{DudKon:2012}, we need to obtain a few somewhat subtle relationships between 
the operators $\mop(\cdot)$ and $\ind(\cdot)$.
\begin{definition}
Let $X$ be any language, and consider again our fixed 
input-altering transducer $\trt$.
We define the following notation and concepts.
\begin{enumerate}
\item 
$\invtx{X,\trt}=\trt^{-1}\outin\ind(X)$. When $\trt$ is understood we simply write $\invt$ instead of $\invtx{X,\trt}$.
\item
\trt{} is called \emdef{exhaustive}, if $\trt^\cap(X)=\es$, for every language $X$.
\item
\trt{} is called \emdef{smooth}, if $\invt^\cap(\ind(X)) \sse \invt^*(\mop X)$, for every language $X$.
\end{enumerate}
\end{definition}
\pnsn
One verifies that exhaustive $\trti$ implies smooth $\trt$.
\begin{lemma}\label{lemSigma}
Let $X,Y$ be any languages. 
The following statements hold true.
\begin{enumerate}
\item
$\invt(X)=\es=\invt^+(X)$ and $\invt(Y)\sse\ind(X)$.
Also, if $X\sse Y$ then $\ind(Y)\cap\trt^{-1}(X)=\es$
and $\invtx{Y}(A)\sse\invt(A)$ for all languages $A$.
\item
$\mop X=\ind(X)\sm\invt(\ind(X))\>$ and $\>\ind(X)=\mop X\cupdot\invt(\ind(X))$.
\item
If $X\sse Y$ and $i\in\N$ then $\invtY^i(A)\sse\invt^i(B)$ for all languages 
$A,B$ with $A\sse B$.
\item
If $\trt$ is transitive then also $\invt$ is transitive.
\item
$\ind(X)=\invt^*(\mop X)\>\cup\>\invt^\cap(\ind(X))$.
\item
If $X=\mop X$ then  $\ind(X)=X\cup\invt^\cap(\ind(X))$.
\end{enumerate}
\end{lemma}
\begin{proof}
The first two statements follow from the definitions of the operators
$\invt$, $\ind$ and $\mop$.
For example, 
\[
\invt(X)=\trt^{-1}(X)\cap\ind(X)=\trt^{-1}(X)\cap
(\mxl\sm\trt(X)\sm\trt^{-1}(X))=\es.
\]
The \underline{third} statement follows from the second one using induction on $i$.
The \underline{fourth} statement follows when we note that $\trt^{-1}$ must be transitive (hence $\trt^{-2}\sse\trti$) and, for all
words $u,v$, $u\in\invt^2(v)$ implies the existence of a 
word $z$ such that
$u\in\invt(z)$ and $z\in\invt(v)$, which implies
$u\in\trt^{-2}(v)$ and $u\in\ind(X)$.
For the \underline{fifth} statement, first note that 
$
\ind(X)=\mop X\>\cup\invt(\mop X)\>\cup\invt^2(\ind(X)),
$
which implies that, for all $i\in\N_0$,
\begin{equation}\label{eqI}
\ind(X)=\invt^{\le i}(\mop X)\cup\invt^{i+1}(\ind(X)).
\end{equation}
If $w\in \invt^*(\mop X)\cup\invt^\cap(\ind(X))$, one uses
Eq.~(\ref{eqI}) to show that $w\in\ind(X)$. Conversely,
if $w\in \ind(X)$ and $w\notin\invt^\cap(\ind(X))$, then
there is $j\in\N$ such that $w\notin\invt^j(\ind(X))$ and,
by Eq.~(\ref{eqI}) $w\in\invt^{\le j-1}(\mop X)$.
The \underline{last} statement follows from the previous statements, when we note that
\[
\ind(X)=\mop X\cup\invt^+(\mop X)\cup\invt^\cap(\ind(X)).
\]
\end{proof}
The first statement of the next lemma is needed in this section. The rest of the
statements are used in the next section, but we include them here as they concern smooth transducers.
The lemma uses the following notation, for $i\in\N$
\begin{equation}\label{eqD}
\diff{\trt}{i}X=\mop^i X\sm\mop^{i-1}X.
\end{equation}
As before, when $\trt$ is understood, it is omitted in the above notation.
\begin{lemma}\label{lemSigmaSmooth}
Let $i\in\N$, let $X$ be any language and assume that the fixed 
transducer $\trt$ is smooth. Let $L$ be any language satisfying $\ptyt$. The following statements hold true.
\begin{enumerate}
\item
$\ind(X)=\invt^*(\mop X)$ and $\mop X\cap\invt^+(\mop X)=\es$.
\item
$\invtx{\mop^{i-1}L}^j(\mop^i L) = \invtx{\mop^{i-1}L}^j(\diff{}{i}L)$, for all $j\in\N$.
\item
$\ind(\mop^{i} L)\sse \mop^i L\cupdot\invtx{\mop^{i-1} L}^{\ge2}(\diff{}{i} L)
\sse \mop^i L\cup(\trt^{-1})^{\ge2}(\diff{}{i} L)$
\item
$\diff{}{i+1}L\sse\invtx{\mop^{i-1} L}^{\ge2}(\diff{}{i}L)$.
\item
$\invtx{\mop^i L}^{\ge2}(\mop^{i+1} L)\sse
\invtx{\mop^{i-1} L}^{\ge4}(\mop^i L)$
\item
$\invtx{\mop^{i-1} L}^{\ge2}(\mop^i L)
=\invtx{\mop^{i-1} L}^{\ge2}(\diff{}{i} L)
\sse
\invtx{L}^{\ge2i}(\diff{}{1} L)$.
\end{enumerate}
\end{lemma}
\begin{proof}
We use the previous lemma. In particular, as $\trt$ is smooth, we have
\[
\ind(X)=\invt^*(\mop X).
\]
\begin{enumerate}
\item Now,
$
\mop X=\ind(X)\sm\invt(\ind(X))=\invt^*(\mop X)\sm\invt^+(\mop X)=\mop X\sm\invt^+(\mop X).
$
Thus, $\mop X\cap\invt^+(\mop X)=\es$.
\item
As $\invtx{\mop^{i-1}L}^+(\mop^{i-1} L)=\es$, we have
$\invtx{\mop^{i-1}L}(\mop^i L) = 
\invtx{\mop^{i-1}L}(\diff{}{i}L\cup\mop^{i-1} L) = 
\invtx{\mop^{i-1}L}(\diff{}{i}L)\cup\invtx{\mop^{i-1}L}(\mop^{i-1} L) = 
\invtx{\mop^{i-1}L}(\diff{}{i}L)$. The statement now follows.
\item
Using the previous statement we have, 
\pssn
$\ind(\mop^{i} L)\sse \ind(\mop^{i-1} L)=
\mop^i L\cupdot\invtx{\mop^{i-1} L}(\ind(\mop^{i-1} L))=
\mop^i L\cupdot\invtx{\mop^{i-1} L}(\invtx{\mop^{i-1} L}^*(\mop^i L)) = 
\mop^i L\cupdot\invtx{\mop^{i-1} L}^+(\diff{}{i} L)$. 
Also, as $\ind(\mop^i L)\cap\trt^{-1}(\mop^i L)=\es$, we have that 
$\ind(\mop^{i} L)\sse \mop^i L\cupdot \invtx{\mop^{i-1} L}^{\ge2}(\diff{}{i} L)$.
\item
As $\mop^{i+1} L\sse\ind(\mop^i L)$, the statement follows from the previous one.
\item
$\invtx{\mop^i L}^{\ge2}(\mop^{i+1} L)=
\bigcup_{j\ge2}\invtx{\mop^i L}^{j}(\mop^{i+1} L)=$\pssn
$\bigcup_{j\ge2}\Bigl(\invtx{\mop^i L}^{j}\bigl(\mop^{i+1} L\sm\mop^{i} L\bigr)\cup\invtx{\mop^i L}^{j}\bigl(\mop^{i} L\bigr)\Bigr)=$\pssn
$\bigcup_{j\ge2}\Bigl(\invtx{\mop^i L}^{j}\bigl(\mop^{i+1} L\sm\mop^{i} L\bigr)\>\cup\es\>\Bigr)\sse$\pssn
$\bigcup_{j\ge2}\Bigl(\invtx{\mop^i L}^{j}\bigl(\invtx{\mop^{i-1} L}^{\ge2}(\mop^i L)\bigr)\Bigr)=\invtx{\mop^{i-1} L}^{\ge4}(\mop^i L)
$
\item 
First note that, for all $j\in\N$,
\[
\invtx{\mop^{j-1} L}^{\ge4}(\mop^j L)=
\invtx{\mop^{j-1} L}^{2}(\invtx{\mop^{j-1} L}^{\ge2}(\mop^j L))\sse
\invtx{L}^{2}(\invtx{\mop^{j-1} L}^{\ge2}(\mop^j L)).
\]
Then the statement follows when we use induction
on $i$ to show 
$\invtx{\mop^{i-1} L}^{\ge2}(\mop^i L)
\sse
\invtx{L}^{\ge2i}(\mop L)$.
\end{enumerate}
\end{proof}
\begin{lemma}\label{lemMax}
Let $L$ be a language satisfying the property $\ptyt$. The following statements hold true.
\begin{enumerate}
\item
$L$ is $\ptyt$-maximal if and only if $\>\ind(L)\sse L.$
\item
If $L$ is $\ptyt$-maximal then $L=\mop^i L$, for all $i\in\N$.
\end{enumerate}
\end{lemma}
\begin{proof}
For the first statement, following \cite{DudKon:2012} we have that $L$ is $\ptyt$-maximal if and only if,
\[
\mxl\sm(L\cup\trt(L)\cup\trt^{-1}(L))=\es,
\]
if and only if, $\ind(L)\sm L=\es$, if and only if, $\ind(L)\sse L$.
\pnsi
For the second statement,
let $L$ be $\ptyt$-maximal. Then $\ind(L)\sse L$, and the statement follows from
Lemma~\ref{lemXYL}.
\end{proof}

\begin{theorem}\label{thGeneral}
Let $L$ be a language satisfying the property $\ptyt$. 
If $\trt$ is smooth and there is $i\in\N_0$ such that $\mop^{i+1} L=\mop^i L$ then $\mop^i L$ is $\ptyt$-maximal
and contains $L$.
\end{theorem}
\begin{proof}
Assume $\mop^{i+1}L =\mop^i L$ and let $K=\mop^i L$. 
Then $\mop K=K$.
As $\trt$ is smooth, Lemmata~\ref{lemSigma} and~\ref{lemSigmaSmooth} imply 
\[
\ind(K)=\invtK^*(\mop K)=\mop K\cup\invtK^+(\mop K)=K\cup\invtK^+(K)=K.
\]
This implies that $K$ is maximal using Lemma~\ref{lemMax}.
\end{proof}

%
%
The next theorem is a slightly stronger version of a result 
in \cite{VHH:2005}
which states that $L$ is included in the maximal $\mop L$ 
when $\trt$ is
length-decreasing-and-transitive.
\begin{theorem}\label{resApp1}
Let $L$ be a language satisfying the property $\ptyt$. 
If $\trt$ is transitive and smooth then $\mop L$ is $\ptyt$-maximal
and contains $L$.
\end{theorem}
\begin{proof}
Assume $\trt$ is transitive and smooth. Then, $\trt^2\sse\trt$.
By Lemma~\ref{lemMax}, it is sufficient to show  that $\ind(\mop L)\sse\mop L$. First note that 
Lemmata~\ref{lemSigma} and~\ref{lemSigmaSmooth} imply 
$
\ind(L)=\mop L\cup\invtL(\mop L).
$
Then,
\[
\ind(\mop L)\sse\ind(L)=\mop L\cup\invtL(\mop L).
\]
By definition of $\ind(\cdot)$, we have 
$\ind(\mop L)\cap\trt^{-1}(\mop L)=\es$,
which implies that 
$\ind(\mop L)  
\sse\mop L$, as required.
\end{proof}

\begin{example}
The input-altering transducer $\trt_1$ in Fig.~\ref{fig:tr1} 
is transitive and smooth but  neither length-decreasing nor
length-increasing. It is transitive because
of the fact that $\trt_1^2(x)=\es$ for all $x$. This last fact 
implies that $\trt_1^i(x)=\es$ for all $x$ and $i\ge2$, which
implies further that $\trti_1$ is exhaustive and, therefore,
$\trt$ is smooth. 
\begin{figure}[ht!]
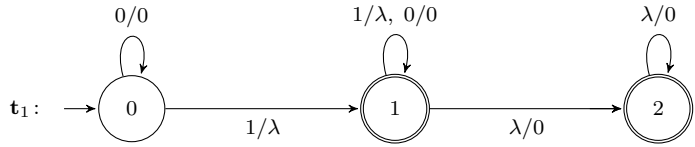

\begin{center}
\begin{transducer}
	\node [state,initial] (q0) {$0$};
	\node [node distance=1cm,left=of q0,anchor=east] {$\trt_{1}\colon$};
	\node [state,accepting,right of=q0] (q1) {$1$};
	\node [state,initial,accepting,right of=q1] (q2) {$2$};
	\path (q0) edge [loop above] node [above] {$0/0$} ()
		(q0) edge node [below] {$1/\e$} (q1)
		(q1) edge [loop above] node [above] {$1/\e,\;0/0$} ()	
		(q1) edge node [below] {$\e/0$} (q2)
		(q2) edge [loop above] node [above] {$\e/0$} ();
\end{transducer}
\parbox{4.3in}{\caption{An example of an input-altering transducer that is smooth and transitive, but neither length-decreasing nor
length-increasing.}
\label{fig:tr1}}
\end{center}
\end{figure}
\end{example}

\section{Input-decreasing Transducer Properties}\label{sec:decr}
In this section, we consider a fixed, but arbitrary, total order $\ord$ 
on the set $\al^*$ of all words. Then, every word $w$ has
a \emdef{position} $\pos(w)$  with respect to that order, starting from position 0.
Moreover, $v\ord w$ implies $v\not=w$, for any $v,w\in\al^*$.
We also consider a fixed, but arbitrary, transducer $\trt$  such that
\[
y\in\trt(x)\>\hbox{ implies }\> y\ord x
\]
for all words $x,y$. Any transducer satisfying the above condition is called an \emdef{input-decreasing transducer}.

\begin{definition}\label{defIDTP}
An \emdef{input-decreasing transducer property} is a property that is equal to $\pty_\trt$ for some input-decreasing transducer $\trt$. 
\end{definition}
\begin{remark}\label{rem:IDTP}
Input-decreasing transducer properties are closed under intersection, as $(\trt\lor\trs)$ is input-decreasing when both $\trt$ and $\trs$ are. 
\end{remark}
\begin{remark}\label{rem:IDTP2}
For  any binary relation $\rho$, let 
$$\rho_{\ord}=\{(x,y)\in\rho\mid x\ord y\},$$
and assume that $\rho_{\ord}$ can be realized by an input-decreasing transducer $\trt$. One
verifies that a language $L$ is $\rho$-independent
if and only if it satisfies $\ptyt$ (that is, $L$ is $\trt$-independent).
\end{remark}
\begin{lemma}\label{lemOrd}
Consider the fixed input-decreasing transducer $\trt$.
The following statements hold true.
\begin{enumerate}
\item
\trt\ is input-altering
\item 
$\trt(x)$ is finite, for all words $x$.
\item
$\trt$ is smooth.
\end{enumerate}
\end{lemma}
\begin{proof}
The first two statements follow from the assumption that $\trt$
is input-decreasing.
For the last statement, we show that 
$\trti$ is exhaustive using contradiction. So
assume there is a language $X$ and a word $w$ such that $w\in(\trti)^\cap(X)$. Let $p$ be the position of $w$ 
with respect to the total order $\ord$. 
As $w\in(\trti)^{p+1}(X)$, there are words 
$x_0,x_1,\ldots,x_p\in X$ such that
\[
x_1\in\trti(x_0),\>
x_2\in\trti(x_1),\ldots,\>x_p\in\trti(x_{p-1}),\>
w\in\trti(x_p).
\]
Then, $x_0\ord x_1\ord\cdots\ord x_p\ord w$, which implies that the position of $w$ is greater than $p$, a contradiction.
\end{proof}
\begin{theorem}
Assume that $\trt$ is input-decreasing. If a language
$L$ satisfies $\ptyt$ then the language $\mop^*L$
is $\ptyt$-maximal and contains $L$.
\end{theorem}
\begin{proof}
That $\mop^*L$ contains $L$ follows from Lemma~\ref{lemXYL}.
Also, using the same lemma it follows that no two words $u,v$ in 
$\mop^*L$ are related via $\rel(\trt)$ and, therefore, $\mop^*L$ satisfies $\ptyt$. To show that $\mop^*L$ is maximal we pick any word $w\in\ind(\mop^*L)$ and show that $w\in\mop^*L$---see Lemma~\ref{lemMax}. By the definition of the operator $\ind$, we have that 
\[
w\in\bigcap_{i\ge0}\bigl(\mxl\sm(\trt\lor\trt^{-1})(\mop^i L)\bigr)=
\bigcap_{i\ge0}\ind(\mop^i L).
\]
Then by Lemma~\ref{lemSigmaSmooth}, for every nonnegative integer $i$, we have that 
\[
w\in\invtx{\mop^i L}^*(\mop^{i+1} L)=(\mop^{i+1}L)\;\cup\;\invtx{\mop^i L}(\mop^{i+1}L)\;
\cup\;\invtx{\mop^i L}^{\ge2}(\mop^{i+1}L).
\]
Now let $i=\lfloor\pos(w)/2\rfloor$. As $w\in (\mxl\sm(\trt\lor\trt^{-1})(\mop^{i+1} L))$, we have that $w\notin\invtx{\mop^i L}(\mop^{i+1}L)$. Also, by Lemma~\ref{lemSigmaSmooth}(6), any 
$u\in\invtx{\mop^i L}^{\ge2}(\mop^{i+1}L)$ must have $\pos(u)\ge 2i+2$ and, therefore, $w\notin \invtx{\mop^i L}^{\ge2}(\mop^{i+1}L)$. Hence, $w\in (\mop^{i+1}L)\sse\mop^*L$, as required.
\end{proof}
In the above theorem, the premise that $\trt$ be input-decreasing is essential. This is shown next with examples.
\begin{example}\label{ex2}
Let $\tpx$ and $\tsx$ be the  input-decreasing transducers describing, respectively, prefix codes and suffix codes.
\begin{figure}[ht!]
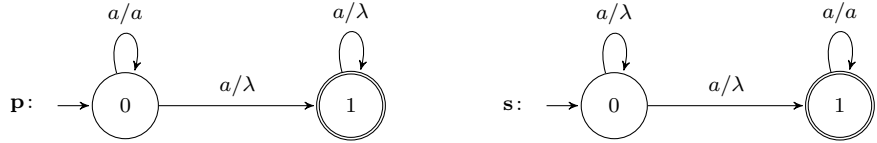

\begin{center}
\begin{transducer}[node distance=3cm]
	\node [state,initial] (q0) {$0$};
	\node [node distance=1cm,left=of q0,anchor=east] {$\tpx\colon$};
	\node [state,accepting,right of=q0] (q1) {$1$};
	\node [state,initial,right of=q1,node distance=3.5cm] (q2) {$0$};
	\node [node distance=1cm,left=of q2,anchor=east] {$\tsx\colon$};
	\node [state,accepting,right of=q2] (q3) {$1$};
	\path (q0) edge [loop above] node [above] {$a/a$} ()
		(q0) edge node [above] {$a/\ew$} (q1)
		(q1) edge [loop above] node [above] {$a/\ew$} ()	
		(q2) edge node [above] {$a/\ew$} (q3)
		(q2) edge [loop above] node [above] {$a/\ew$} ()
		(q3) edge [loop above] node [above] {$a/a$} ();
\end{transducer}
\parbox{4in}{\caption{The left transducer  describes  prefix codes: on input $x$ it
outputs any proper prefix of $x$. The right transducer 
describes suffix codes. Both transducers are input-decreasing. \underline{Note:} in this and the following transducer figures, an arrow
with label $a/a$ represents a set of edges with labels $a/a$ for all $a\in\al$; and similarly for an arrow
with label  $a/\ew$. 
An arrow with label $a/a'$ represents a set of edges with labels $a/a'$ for all $a,a'\in\al$ with $a\not=a'$.}
\label{fig:pxsx}}
\end{center}
\end{figure}
Let $\tbx=(\tpx\lor\tsx)$ be the input-decreasing
transducer describing bifix codes. Then, for any bifix code $L$, the language $\mop_{\tbx}^*L$ is a maximal bifix code containing $L$. On the other hand, if we describe bifix codes using \emshort{any} of the three transducers
\[
(\tpx^{-1}\lor\tsx),\>\>
(\tpx\lor\tsx^{-1}),\>\> 
(\tpx^{-1}\lor\tsx^{-1}),
\]
then the theorem does not hold---although all three are input-altering, none of them is input-decreasing. For example, with $\trt$ being any of those three transducers, 
and for $\al=\{0,1\}$,
we have $\mop(001)=001$, hence
$\mop^*(001)=001$; that is on input 001, the iterated max-min operator
returns  001 itself, which is not maximal bifix---here we have
used FAdo~\cite{Fado} for computations on automata and transducers.
\end{example}
\begin{remark}\label{rem:regular}
If the language $L$ is regular then also the language
$\mop^i L$ is regular, for all $i\in\N$. This follows by the definition of $\mop$ and the standard closure properties of regular languages. In particular, an automaton accepting $\mop^i L$ can be effectively computed
from any automaton accepting $L$. Thus, if $\trt$ is such that $\mop^*=\mop^i$  
for some index $i$, then a maximal regular embedding of any given regular $L$ can be effectively computed.
\end{remark}
\begin{theorem}
Assume that $\mxl$ is finite and $\trt$ is 
input-decreasing. If a language
$L\sse\mxl$ satisfies  property $\ptyt$
then there is $i\in\N$ such that the 
language
$\mop^{i}L$ is regular, $\pty_\trt$-maximal and contains $L$.
\end{theorem}
\begin{proof}
Let $\mxl=\{w_1,\ldots,w_n\}$ for some $n\in\N$.
As each $\trt(w_i)$ is finite there is $p_i\in\N$ such that $\trt^{p_i}(w_i)=\es$. Hence, $\trt^p(\mxl)=\es$, where
$p=\max_i\{p_i\}$. Then we have that 
$(\trt^p)^{-1}(\mxl)=\es$, and also $(\trt^{-p})(\mxl)=\es$,
by Remark~\ref{rem1}. Then Lemma~\ref{lemSigmaSmooth} implies\pssi
$\mop^{1+\lceil p/2\rceil}(L)\sm\mop^{\lceil p/2\rceil}(L)\sse
\invtx{L}^{\ge2\lceil p/2\rceil}(\mop L)\sse
\invtx{L}^{\ge p}(\mop L)\sse\es$. 
\pssn
Hence,
$\mop^{\lceil p/2\rceil}(L)$ 
is $\pty_\trt$-maximal by Theorem~\ref{thGeneral}.
\end{proof}
As before, the premise that $\trt$ be input-decreasing is essential. This is shown next with an example.
\begin{example}\label{ex2b}
Let $\tsub1$ be the  input-altering transducer (shown below) describing 1-substitution-detecting languages. A language $L$ is 
\emdef{$k$-substitution detecting} if no $L$-word can result into another $L$-word using up to $k$ symbol substitutions (one substitution = one symbol replaced with another one). 
\begin{figure}[ht!]
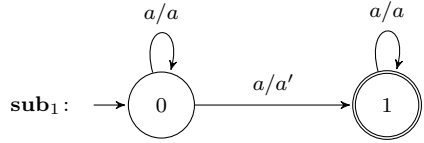

\begin{center}
\begin{transducer}[node distance=3cm]
	\node [state,initial] (q0) {$0$};
	\node [node distance=1cm,left=of q0,anchor=east] {$\tsub1\colon$};
	\node [state,accepting,right of=q0] (q1) {$1$};
	\path (q0) edge [loop above] node [above] {$a/a$} ()
		(q0) edge node [above] {$a/a'$} (q1)
		(q1) edge [loop above] node [above] {$a/a$} ();
\end{transducer}
\parbox{4in}{\caption{This transducer  describes  1-substitution-detecting languages---see caption of
the previous figure for explanations on transducer diagrams. It is input-altering 
but not input-decreasing.}
\label{fig:subd}}
\end{center}
\end{figure}
The transducer is not input-decreasing, as $0\in\tsub1(1)$ and $1\in\tsub1(0)$. Moreover for $\trt=\tsub1$, we have that $\mop (0000) = 0000$ and, hence, $\mop^*(0000)=0000$, which is not maximal 1-substitution-detecting.
\end{example}
%
%
%
\section{Examples and further observations}\label{sec:ex}
In this section we use the standard quasi-lexicographic (or radix) total order on all words over $\{0,1,\ldots, q-1\}$, for some integer $q\ge2$. Thus, $u\ord v$ means that, either $u$ is shorter, or $u$ and $v$ are of the same length and, for the first position in which they differ,  the symbol of $u$ at that position is smaller than that of $v$. 
All the examples presented below have been
confirmed using the well-maintained Python 
package FAdo \cite{Fado}, which was recently updated to 
include a module on codes described by input-altering transducers \cite{KMMR:2015}.
\pmsn
In our examples below we use notation of regular expressions. For instance, $01^*0(0+1)$ denotes the language 
$\{01^i0\mid i\in\N_0\}\{0,1\}$.
\begin{example}\label{ex3}
Let $\mxl=\{0,1\}^*$ and $\trt=\tbx$ =  the input-decreasing transducer 
describing bifix codes. We have that
\[
\mop(001) = \{001,000,10,11\},\>\hbox{ and }\> 
\mop^2(001) = 01^*0(0+1)+10+11
\]
which  is maximal. Again, we have
\begin{align*}
\mop((0+1)^30) = (0+1)^4,\>\hbox{ which is maximal, }\\
\mop^2((0+1)^311) = (0+1)^3(0+10^*1),\>\hbox{ which is maximal. }
\end{align*}
The last code above is the reverse of a code in \cite{BePeRe:2009}, which is called there
reversible Golomb-Rice code.
Finally, note that $\mop^5$ on $11111$ 
generates a maximal bifix code.
\end{example}
For the next examples we use the two transducers
shown below over the binary alphabet $\{0,1\}$.
\begin{figure}[ht!]
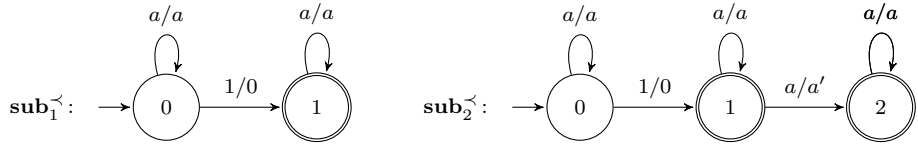

\begin{center}
\begin{transducer}[node distance=2cm]
	\node [state,initial] (q0) {$0$};
	\node [node distance=1cm,left=of q0,anchor=east] {$\tdsub{1}\colon$};
	\node [state,accepting,right of=q0] (q1) {$1$};
	\node [state,initial,right of=q1,node distance=3.5cm] (q2) {$0$};
	\node [node distance=1cm,left=of q2,anchor=east] {$\tdsub2\colon$};
	\node [state,accepting,right of=q2] (q3) {$1$};
	\node [state,accepting,right of=q3] (q4) {$2$};
	\path (q0) edge [loop above] node [above] {$a/a$} ()
		(q0) edge node [above] {$1/0$} (q1)
		(q1) edge [loop above] node [above] {$a/a$} ()	
		(q2) edge [loop above] node [above] {$a/a$} ()
		(q2) edge node [above] {$1/0$} (q3)
		(q3) edge [loop above] node [above] {$a/a$} 	    	(q3) edge node [above] {$a/a'$} (q4)
		(q4) edge [loop above] node [above] {$a/a$} 			(q3) edge [loop above] node [above] {$a/a$} ();
\end{transducer}
\parbox{4in}{\caption{On input $x$, the left transducer
outputs any word resulting by substituting exactly one 1 in $x$ with a 0. Note that $(\tdsub1)\lor(\tdsub1)^{-1}$
is equal to the transducer $\tsub1$
and, therefore, $\tdsub1$ 
describes  the 1-substitution-detecting languages over $\{0,1\}$.  The right transducer 
describes the 2-substitution-detecting languages. Both transducers are input-decreasing. }
\label{fig:ed}}
\end{center}
\end{figure}
\begin{example}\label{ex4}
Let $\mxl=\{0,1\}^5$ and $\trt=\tdsub1$ =  the input-decreasing transducer 
describing 1-substitution-detecting languages. 
We have that
\[
\mop^3(01111) = 
\{w\in\{0,1\}^5\mid \> \hbox{$w$'s count of 1s is even}\}
\]
This code is maximal and known as the even-parity code of length 5, which constitutes a vector space of dimension 4 consisting of $2^4$ codewords.
\end{example}
\begin{example}\label{ex5}
Let $\mxl=\{0,1\}^7$ and $\trt=\tdsub2$ =  the input-decreasing transducer 
describing 2-substitution-detecting languages. 
We have that
\begin{align*}
\mop^6(1111111) = &
\{0000000,1001011,0101010,1100001,0011001,1010010,\\
&\>0110011,1111000,   
  0000111,1001100,0101101,1100110,\\
&\>0011110,1010101,0110100,1111111 \}
\end{align*}
This code is the reverse of the Hamming code of 
length 7 \cite{Ham:1950}.
It is maximal and constitutes a vector space of dimension 4 consisting of $2^4$ codewords. It is also 1-substitution-\emshort{correcting}.
\end{example}
In the next example we use the  input-decreasing transducer
$\tdid2$ shown in Fig~\ref{fig:ed2}, which describes 2-insertion-deletion-detecting languages. A language $L$ is 
\emdef{$k$-insertion-deletion detecting} if no $L$-word can result into another $L$-word using a total of up to $k$ symbol insertions/deletions. 
The challenge in designing the transducer is to make sure that two insertion-deletion errors on some input word $x$ do cause the resulting word to be different from $x$ and smaller than $x$. Moreover, the transducer $(\tdid2)\lor(\tdid2)^{-1}$ is such that, on any input word $x$, one or two insertion/deletion errors are applied resulting into a word not equal to $x$. The main idea is that $\tdid2$ applies an insertion and a deletion in two ways: (i) A deletion of 1 immediately followed by either a 0 not changed or an inserted 0. This is justified, as deleting a 1 in a run of 1s has the same effect as deleting the last 1 of that run; (ii) An insertion of a 0 immediately followed by either a 1 not changed or a deleted 1. Again, this is justified as inserting a 0 in a run of 0s has the same effect as inserting the 0 at the end of that run.

\begin{figure}[ht!]
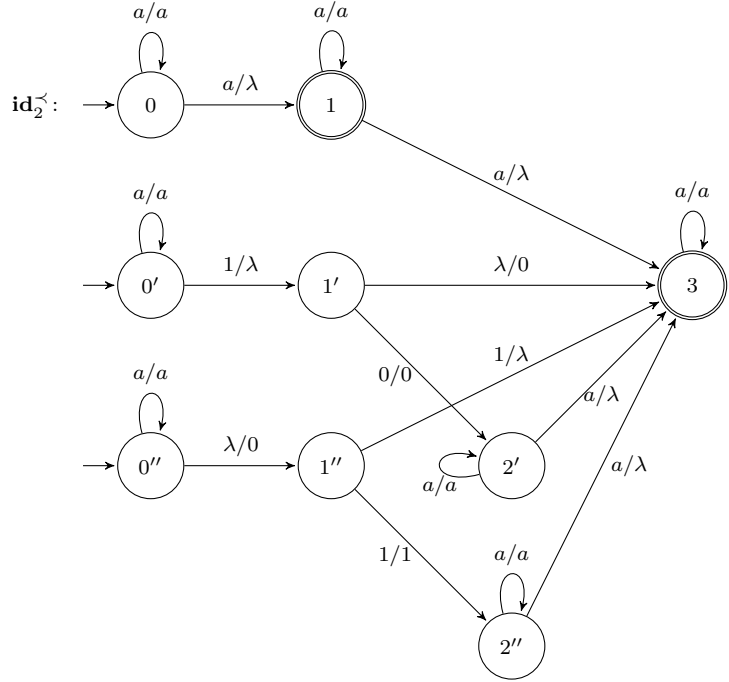

\begin{center}
\begin{transducer}[node distance=2.4cm]
	\node [state,initial] (q0a) {$0$};
	\node [node distance=1cm,left=of q0a,anchor=east] {$\tdid{2}\colon$};
	\node [state,initial,below of=q0a] (q0b) {$0'$};
	\node [state,initial,below of=q0b] (q0c) {$0''$};
	\node [state,accepting,right of=q0a] (q1a) {$1$};
	\node [state,right of=q0b] (q1b) {$1'$};
	\node [state,right of=q0c] (q1c) {$1''$};
	
	\node [right of=q1b] (nb) {};
	\node [right of=q1c] (nc) {};
	
	\node [state,below of=nb] (q2b) {$2'$};
	\node [state,below of=nc] (q2c) {$2''$};
	
	\node [state,accepting,right of=nb] (q3) {$3$};
	
	\path (q0a) edge [loop above] node [above] {$a/a$} ()
		(q0a) edge node [above] {$a/\ew$} (q1a)
		(q1a) edge [loop above] node [above] {$a/a$} ()
		(q1a) edge node [above] {$a/\ew$} (q3)

        (q0b) edge [loop above] node [above] {$a/a$} ()
		(q0b) edge node [above] {$1/\ew$} (q1b)
		(q1b) edge node [above] {$\ew/0$} (q3)
		(q1b) edge node [left] {0/0} (q2b)
		(q2b) edge [loop left] node [below] {$a/a$} ()
		(q2b) edge node [below] {$a/\ew$} (q3)

        (q0c) edge [loop above] node [above] {$a/a$} ()
		(q0c) edge node [above] {$\ew$/0} (q1c)
		(q1c) edge node [above] {$1/\ew$} (q3)
		(q1c) edge node [left] {1/1} (q2c)
		(q2c) edge [loop above] node [above] {$a/a$} ()
		(q2c) edge node [right] {$a/\ew$} (q3)
        (q3) edge [loop above] node [above] {$a/a$} ();
\end{transducer}
\parbox{4in}{\caption{This is an input-decreasing transducer 
describing  the 2-insertion-deletion-detecting languages over $\{0,1\}$.  
}
\label{fig:ed2}}
\end{center}
\end{figure}
\begin{example}\label{ex6}
Let $\mxl=\{0,1\}^6$ and $\trt=\tdid2$ =  the input-decreasing transducer 
describing 2-insertion-deletion-detecting languages. 
We have that
\begin{align*}
\mop^5(001011) = \{& 000000,001011,001100,010001,011101,\\
&
101010,110000,110011,111100,111111\}
\end{align*}
This code is maximal and consists of 10 codewords. 
Any 2-insertion-deletion-detecting code of fixed length
has a Levenshtein distance greater than 2 and, therefore,
it is also 1-insertion-deletion-\emshort{correcting}.
We note that the Levenshtein  1-insertion-deletion-correcting code of length 6 in \cite{Levenshtein:66:en} is maximal and consists of 10 codewords
as well.
\end{example}
\begin{example}\label{ex7}
Let $\mxl=\{0,1\}^{\le6}$ and $\trt=\tpx\lor\tdsub1$ =  the input-decreasing transducer describing languages over $\{0,1\}$ that are both 
1-substitution-detecting and prefix codes. 
We have that
\[
\mop^5(111) = 
\{0,10,111,1100,11010,110110\}
\]
is maximal (relative to $\{0,1\}^{\le6}$). 
\end{example}
The next result shows an example of an input-decreasing transducer and language on which the $\mop^*$ does not 
converge finitely. First we establish the following lemma.
\begin{lemma}\label{lem:ex}
Let $n\in\N_0$, let $\al=\{0,1\}$, let $L_n=\{1,00,010,\ldots,01^{n-1}0\}$,
and let $\trt=\tpx\lor\tdsub1$ =  the input-decreasing transducer describing the languages  that are both 
1-substitution-detecting and prefix codes. 
We have that 
\[
\ind(L_n) = L_n\cup01^n\al^+.
\]
\end{lemma}
\begin{proof}
Recall that $\tdsub1$ substitutes exactly one 1 with a 0, and 
hence, $\tdsubi1$ substitutes exactly one 0 with a 1. 
Thus, 
\[
\tdsubi1(L_n)=(10+110+\cdots+1^n0)+(01+011+\cdots+01^n).
\]
We use the notation $(x)^{1/0}$ to denote the set of all words that result by substituting exactly one 1 with a 0 in the word $x$. Thus,
\[
\tdsub1(L_n)=0+0(1)^{1/0}0+\cdots+0(1^{n-1})^{1/0}0.
\]
Using the definition of $\ind(\cdot)$ and the fact
$\al^*=\ew+0+1+0\al^++1\al^+$, we have
\begin{align*}
\ind(L_n) =& \>\al^*\sm\tpx(L_n)\sm L_n\al^+\sm\tdsub1(L_n)\sm\tdsubi1(L_n) = \\
&1+0\al^+\sm00\al^+\sm\cdots\sm01^{n-1}0\al^+\sm01\sm01^2\sm\cdots\sm01^n \\
&\>\sm000\sm0(11)^{1/0}0\sm\cdots\sm0(1^{n-1})^{1/0}0=\\
&(1+00)+\bigl(01\al^+\sm010\al^+\cdots\sm01^{n-1}0\al^+\sm01^2\sm\cdots\sm01^n\\
&\>\sm000\sm0(11)^{1/0}0\sm\cdots\sm0(1^{n-1})^{1/0}0\bigr)=\\
&(1+00+010)+\bigl(011\al^+\sm01^20\al^+\cdots\sm01^{n-1}0\al^+\sm01^3\sm\cdots\sm01^n\\
&\>\sm000\sm0(111)^{1/0}0\sm\cdots\sm0(1^{n-1})^{1/0}0\bigr)=
\cdots\cdots=\\
&(1+00+010+\cdots+01^{n-2}0)\\
&\>+
(01^{n-1}\al^+\sm01^{n-1}0\al^+\sm01^n\sm0(1^{n-1})^{1/0}0)=\\
&L_n+(01^n\al^+\sm0(1^{n-1})^{1/0}0)=L_n+01^n\al^+,
\end{align*}
as required.
\end{proof}
\begin{theorem}\label{th:infinite}
Let $\mxl=\al^*$ and $\al=\{0,1\}$. There is an input-decreasing transducer 
$\trt$ such that $(\mop^*1)$ does \emshort{not} converge finitely.
\end{theorem}
\begin{proof}
We consider the notation in the above lemma, and we use induction on $n\in\N_0$ to show that
\[
\mop^n1=\{1,00,010,\ldots,01^{n-1}0\}.
\]
The statement holds for $n=0$. Assume it holds for some $n$, as displayed above, and consider calculating 
$(\mop^{n+1}1)$. Using the definition of $\ind(\cdot)$ we have that
\[
(\mop^{n+1}1)= I(\mop^n1)\sm \trti(I(\mop^n1)) =
\ind(\mop^n 1)\sm\trti(\ind(\mop^n 1)\sm\mop^n1),
\]
where we have used the fact $\trti(\mop^n1)\cap\mop^n1=\es$
as $\mop^n1$ satisfies $\pty_{\trti}$. Now using the above lemma we have that
\[
\mop^{n+1}1=
(\mop^n1)\>\cup\>01^n\al^+\>\>\sm\trti(01^n\al^+).
\]
Using the definition of $\trt$ one verifies that 
$\trt(\mop^n1)\cap01^n\al^+=\es$ and that the only
element of $01^n\al^+$ that does not belong to
$\trti(01^n\al^+)$ is $01^n0$. Then it follows that 
$(\mop^{n+1}1)=(\mop^n1)+01^n0$, as required.
\end{proof}
\begin{example}
In \cite{JiSe:2015} the authors consider computing a maximal prefix code that is a subset of a given regular language $L$.
We can replace prefix code with $\trt$-independent language, for any suitable input-decreasing transducer $\trt$, and
approach this generalized problem by first computing 
$X=L\sm\trti(L)$,
which is always $\trt$-independent, and then use the iterated max-min operator to embed $X$ into a maximal $\trt$-independent language relative to $\mxl=L$. The work in \cite{JiSe:2015} provides further results on maximal prefix codes that can possibly be extended to certain $\trt$-independences---see also the last section for further comments.   
\end{example}
%
%
%
\section{Conclusion}\label{sec:last}
We have shown that when an independence property is
described by an input-decreasing transducer $\trt$,
then the max-min operator $\mop_\trt$ can be iterated 
on any language to produce a maximal embedding.
This approach works for many natural independence properties from
both the noiseless and noisy domains of coding theory, 
as well as for any combinations of such properties. We conclude with a few directions for future research.
\begin{itemize}
\item 
Find out whether, for any given regular bifix code, 
the max-min operator converges  finitely. 
We believe that the answer here is yes.
\item
Find out whether any regular maximal $\trt$-independent language
is the result of applying the iterated max-min operator on some 
initial finite language. This might be true for some cases of $\trt$. Related results of this type exist in \cite{Lam:2000,Lam:2001}.
\item
Study the behaviour of $\mop^*$ on various finite languages,
in particular on singleton languages $\{w\}$.
In this setting, we can talk about the code generated by $w$.
We note that many substitution-detecting codes 
(CRC codes in particular) are generated from a single word, which in fact
is represented by a polynomial~\cite{LinCo:2004}. 
\item
Explore the quality of the maximal languages generated by
$\mop^*$.
In terms of information theory, quality could be the average word length, 
or the efficiency of encoding information, for instance.
In terms of automaton theory, investigate the state complexity of the
regular maximal languages in terms of the state complexity of the initial language, for various cases of $\trt$. 
A study of this type for prefix codes can be found in \cite{JiSe:2015}.
\item
Find out whether the following problem is computable:
given any input-altering transducer $\trt$, return (if possible) an input-decreasing transducer $\trx$ such that 
$\trt\lor\trti=\trx\lor\trx^{-1}$---that is, 
$\pty_\trt=\pty_\trx$. 
\end{itemize}

\bibliographystyle{plain}
\bibliography{refs}

\end{document}